\documentclass[a4paper, 11pt]{article}
\usepackage[top=1 in,bottom=1 in,left=1.0 in,right=1.0 in]{geometry}
\usepackage{amsmath,amsthm}
\usepackage{amssymb}
\usepackage{cite}
\usepackage{CJK}
\usepackage[colorlinks,linkcolor=blue,citecolor=blue]{hyperref}
\numberwithin{equation}{section}
\usepackage{graphicx}
\usepackage{epstopdf}
\usepackage{color}
\usepackage{subcaption}

\newtheorem{prop}{Proposition}[section]

\begin{document}

\title{Cauchy matrix approach to novel extended semi-discrete KP-type systems}

\author{Hong-juan Tian$^{1,2,3}$,~Abdselam Silem$^{4}$\footnote{Corresponding author. Email: silem.a@zjut.edu.cn} \\
{\small$^1$   College of  Computer and Information Engineering, Henan Normal  University, Xinxiang 453007, China }\\
{\small$^2$School of  Physics, Henan Normal  University, Xinxiang 453007, PR China}\\
{\small$^3$ Engineering Lab of Intelligence Business and Internet of Things, Henan Province}\\
{\small$^4$ Department of Mathematics, Zhejiang University of Technology, Hangzhou 310014, China}
}
\date{\today}

\maketitle

\begin{abstract}
Two novel extended semi-discrete KP-type systems, namely partial differential-difference systems with one continuous and two discrete variables, are investigated. Introducing an arbitrary function into the Cauchy matrix function or the plane wave factor allows the implementation of extended integrable systems within the Cauchy matrix approach. We introduce the bilinear $D\Delta^2$KP system, the extended $D\Delta^2$pKP, $D\Delta^2$pmKP, and $D\Delta^2$SKP systems, all of which are based on the Cauchy matrix approach. This results in a diversity of solutions for these extended systems as contrasted to the usual multiple soliton solutions.


\vskip 5pt
\noindent
\textbf{Key Words:} Squared eigenfunction, Semi-discrete KP-type system,  Generalized Cauchy matrix approach, Exact solutions. \\

\end{abstract}

\section{Introduction}\label{sec-1}
Over the past decades, the extensive investigations of the theory of discrete integrable systems \cite{Date-JPSJ-1982,Miwa-1981,PLA-1984,DS-book-2003,DIS-2016-book} have resulted in numerous significant accomplishments within this domain. For instance, integrable semi-discrete models \cite{FW-Nonl-2013,OL-1988-Waveguide,PRL-1998-Waveguide,Tamizhmani-CSF-1997,CK-JNMP-2017,GLW-PA-1988-Toda} are a fundamental subject in the integrable systems field. It is worth mentioning that these equations can serve as the superposition formula for B$\ddot{a}$cklund transformations of their corresponding continuous equations. Moreover, they also possess a significant mathematical value due to their underlying rich algebraic structure\cite{IMRN-2015-Novikov,TMP-2023-classification}. There are numerous methods for searching for integrable discretization of differential equations\cite{PA-1988,TMP-1999-Adler}, with one particularly significant approach being the construction of transformations between solutions \cite{PA-1988}.
The core concept of this discretization method is the establishment of discrete linear dispersion relations \cite{FW-1C2DKP-2021}.
The Cauchy matrix technique, introduced by Nijhoff et al., was initially proposed to study the soliton solutions of the lattice equations \cite{Nij-JPA-2009}. Motivated by these works, a  generalized Cauchy matrix approach\cite{ZZ-SAMP-2013,JNMP-2014} was developed, and has been used to derive several kinds of solutions for continuous and discrete integrable systems. The approach employs the Sylvester equation, which is associated with certain evolutionary relationships between the vectors $\mathbf{r}$ and $\mathbf{s}$. These relationships form the determining equation set (DES), where continuous and discrete KP equations \cite{CNSNS-2013} are governed entirely by the product of the plane wave factors.

In this paper, we propose two novel semi-discrete KP-type integrable systems, they are related to the following differential-difference system (involving one continuous and two discrete variables)
\begin{equation}\label{2dkp-1}
  \partial_x(\widehat{\mathrm{u}}-\widetilde{\mathrm{u}})=(\widehat{\mathrm{u}}-\widetilde{\mathrm{u}})(\mathrm{u}+\widehat{\widetilde{\mathrm{u}}}-\widehat{\mathrm{u}}-\widetilde{\mathrm{u}}).
\end{equation}
Here, we consider the potential function $\mathrm{u}=\mathrm{u}(x,n:p,m:q)$ and adopt notations of forward and backward shifts as follows:
\begin{subequations}
\begin{align}
   &&\widetilde{\mathrm{u}} \doteq \mathrm{u}(x,n+1,m),\widehat{\mathrm{u}}\doteq \mathrm{u}(x,n,m+1),\widehat{\widetilde{\mathrm{u}}} \doteq \mathrm{u}(x,n+1,m+1) \\
  &&\underset{\widetilde{}}{\mathrm{u}} \doteq \mathrm{u}(x,n-1,m),\underset{\widehat{}}{\mathrm{u}}\doteq \mathrm{u}(x,n,m-1),\underset{\widehat{\widetilde{}}}{\mathrm{u}} \doteq \mathrm{u}(x,n-1,m-1)
 \end{align}
\end{subequations}
Equation \eqref{2dkp-1} was implicitly proposed in \cite{Nij-1985} and appeared earlier in the above form as a master equation for generating the discrete Calogero–Moser model (see  Ref.\cite{Nij-PLA-1994}). Extended integrable systems can be generated by making modifications to the plane wave factor or, alternatively, by adding an arbitrary function to the Cauchy function matrix function \cite{CTP-KP-2020,AML-YO-2020,AML-Tian-2022}.

The paper is organized as follows. In Section.\ref{sec-2}, we briefly introduce the infinite matrix formalism of the Cauchy matrix approach. Next, in Section.\ref{sec-3} we demonstrate the efficiency of the Cauchy matrix approach in deriving the novel semi-discrete KP-type equations and the two types of extended integrable systems. Section.\ref{sec-4} provides the explicit formula for all possible vectors $\mathbf{r}, \mathbf{s}$ and the matrix $\mathbf{M}$. These formulas yield a variety of solutions compared to the typical $N$-soliton solutions (NSS) for these extended KP-type systems.

\section{Cauchy Matrix Approach}\label{sec-2}
In addition to deriving soliton solutions, The Cauchy matrix \cite{JNMP-2014}  approach also proved powerful in constructing new equations and their features.
For the three dimensional systems, the Cauchy-type matrix
  $M=(m_{i,j})_{N\times N}, m_{i,j}=\frac{r_is_j}{k_i+l_j}$  is introduced and satisfies the Sylvester equation
\begin{equation}\label{Syl-2}
\mathbf{KM}+\mathbf{ML}=\mathbf{rs}^T,
\end{equation}
where $\mathbf{s}=(s_1,\cdots,s_N)^T$ and
$(r_1,\cdots,r_N)^T$ are  constant vectors.
Assume that $\mathcal{E}(\mathbf{K})\bigcap \mathcal{E}(-\mathbf{L})=\varnothing$,
where $\mathcal{E}(\mathbf{L})$ and  $\mathcal{E}(-\mathbf{K})$ are
the eigenvalue set of $\mathbf{K}$ and $-\mathbf{L}$, respectively.
This will guarantee that the Sylvester equation \eqref{Syl-2} for a given $(\mathbf{L}, \mathbf{K}, \mathbf{r}, \mathbf{s})$
is solvable and has a unique solution $\mathbf{M}$. Scalar functions $\mathbf{S}^{(i,j)}=\mathbf{s}^{T}\mathbf{L}^j(\mathbf{I}+\mathbf{M})^{-1}\mathbf{K}^i\mathbf{r}$ meet certain recurrence relations,
and some relations have closed forms that lead to integrable systems. Note that the generalized Cauchy matrix approach is characterized by the Sylvester equation \eqref{Syl-2} that includes an unknown matrix $\mathbf{M}$ and constant matrices $\mathbf{K,L}$.  Motivated
by the approach's concept, dispersion relations will be imposed on $\mathbf{r}$ and $\mathbf{s}$ or on $\mathbf{M}$. Moreover, we will further concentrate on the evolution of the scalar function $\mathbf{S}^{(i,j)}$. It has the potential to not only integrate the Cauchy matrix technique in discrete and continuous integrable systems, but also to discover further connections linking the discrete systems with their continuous counterparts.

The key point of discretizing the Cauchy matrix approach is to introduce the discrete linear dispersions. For instance, evolutions of discrete, semi-discrete, and continuous KP hierarchies are fully characterized by the product of the plane wave factors
\begin{equation}\label{PWF-1}
 \rho(k)=\exp{\sum^{\infty}_{j=0} k^jx_j}\prod^{\infty}_{i=1}(p_i+k)^{n_i},   \sigma(l)=\exp{(-\sum^{\infty}_{j=0} l^jx_j})\prod^{\infty}_{i=1}(p_i-l)^{-n_i},
\end{equation}
whith $k$ and $l$ being two distinct spectral parameters. Note that the factors related to  $n_j$ can be seen
as discretisations of the exponents of $x_j$\cite{Miwa-1981}.
To clarify, when we substitute the continuous independent variables $x_j$ with discrete variables $n_j$ according to equation \eqref{PWF-1}, a series of equations of different types (continuous, semi-discrete, full-discrete) emerge, successively. The matrix $\mathbf{S}=(S^{(i,j)})_{\infty\times \infty}$ reads
\begin{equation}
 S^{(i,j)} =  \mathbf{s}^{\mathrm{T}} \mathbf{L}^j (\mathbf{I}+\mathbf{M})^{-1} \mathbf{K}^i \mathbf{r}
 =\mathbf{s}^T \mathbf{L}^j  \mathbf{ u}^{(i)}
 =\mathbf{w}^{(j)} \mathbf{K}^i \mathbf{r},
\label{pKP-S}
\end{equation}
where the auxiliary $N$-th order vectors $\mathbf{ u}^{(i)}$ and $\mathbf{w}^{(j)}$ are defined as
\begin{equation}\label{ui-wj}
 \mathbf{ u}^{(i)}= (\mathbf{I}+\mathbf{M})^{-1} \mathbf{K}^i\mathbf{ r},~~
 \mathbf{w}^{(j)} = \mathbf{s}^T\mathbf{L}^j (\mathbf{I+M})^{-1}.
\end{equation}

\subsection{ Semi-discrete KP-type equations }
Based on the generalized Cauchy matrix method, the D$\Delta^2$pKP equation is firstly constructed. Then modifying the plane wave factor and the Cauchy function matrix via adding an arbitrary function, the two type extended D$\Delta^2$pKP systems were obtained. Taking $\mathbf{r, s, M}$ as functions of $( n, m, x )$, different dispersion relations can be obtained for the same Sylvester equation \eqref{Syl-2}.

\begin{prop}\label{2pKP-cm-1}
 For given Sylvester equation \eqref{Syl-2} and dispersion relations
\begin{subequations}\label{2pKP-r-s}
\begin{align}
&\widetilde{ \mathbf{r}}=(p\mathbf{I}+\mathbf{K})\mathbf{r},~\widehat{ \mathbf{r}}=(q\mathbf{I}+\mathbf{K})\mathbf{r},~\mathbf{r}_x= \mathbf{K}\mathbf{r},\label{2dpKP-n} \\
  &\mathbf{s}=(p\mathbf{I}-\mathbf{L})\widetilde{\mathbf{s}},~ \mathbf{s}=(q\mathbf{I}-\mathbf{L})\widehat{\mathbf{s}},~ \mathbf{ s}_x = \mathbf{L}\mathbf{ s},\label{2dpKP-m}
 \end{align}
\end{subequations}
we can obtain the relations :
\begin{subequations}\label{2pKP-S-nmt}
\begin{align}
S^{(i,j)}_x &= S^{(i,j+1)}+S^{(i+1,j)}-S^{(0,j)}S^{(i,0)},\label{2pKP-S-x}\\
p\widetilde{S}^{(i,j)}  &=\widetilde{S}^{(i,j+1)}+pS^{(i,j)}+S^{(i+1,j)}-\widetilde{S}^{(i,0)}S^{(0,j)},\label{2pKP-S-n}\\
q\widehat{S}^{(i,j)}  &=\widehat{S}^{(i,j+1)}+qS^{(i,j)}+S^{(i+1,j)}-\widehat{S}^{(i,0)}S^{(0,j)}.\label{2pKP-S-m}
\end{align}
\end{subequations}
\end{prop}
\begin{proof}
Based on the uniqueness of the solution and dispersion relation \eqref{2pKP-r-s}; together with the Sylvester equation \eqref{Syl-2}  w.r.t discrete variables $n-,m-$direction; and the derivative w.r.t continuous variable $x$, one can obtain the relations
\begin{equation}\label{2pKP-M-xyt}
\widetilde{ \mathbf{M}}= \mathbf{M}+\mathbf{r}\widetilde{\mathbf{s}}^T,~ \widehat{ \mathbf{M}}= \mathbf{M}+\mathbf{r}\widehat{\mathbf{s}}^T,~ \mathbf{M}_x=\mathbf{rs}^T.
\end{equation}
Then, according to the definition $\mathbf{u}^{(i)},\mathbf{w}^{(i)}\eqref{ui-wj}$ and the \eqref{2pKP-M-xyt}, we have
\begin{subequations}\label{2pKP-u-xyt}
 \begin{align}
   &\mathbf{u}_x^{(i)}= \mathbf{u}^{(i+1)}-\mathbf{u}^{(0)}S^{(i,0)},~ \mathbf{w}_x^{(j)}= \mathbf{w}^{(j+1)}-\mathbf{w}^{(0)}S^{(0,j)},\label{2pKP-uw-x}\\
   &\widetilde{\mathbf{u}}^{(i)}= p\mathbf{u}^{(i)}+ \mathbf{u}^{(i+1)}-\mathbf{u}^{(0)}\widetilde{S}^{(i,0)},~ \mathbf{w}^{(j)}= p\widetilde{\mathbf{w}}^{(j)}- \widetilde{\mathbf{w}}^{(j+1)}+\widetilde{\mathbf{w}}^{(0)}S^{(0,j)},\\
   &\widehat{\mathbf{u}}^{(i)}=  q\mathbf{u}^{(i)}+ \mathbf{u}^{(i+1)}-\mathbf{u}^{(0)}\widehat{S}^{(i,0)},~ \mathbf{w}^{(j)}=  q\widehat{\mathbf{w}}^{(j)}- \widehat{\mathbf{w}}^{(j+1)}+\widehat{\mathbf{w}}^{(0)}S^{(0,j)}.
\end{align}
\end{subequations}
Using  $S^{(i,j)}=\mathbf{s}^T\mathbf{L}^j\mathbf{ u}^{(i)}$, when taking the shift in the $n,m$ directions, and the derivative w.r.t $x$, by means of dispersion relation \eqref{2pKP-r-s} and relations, we obtain  \eqref{2pKP-S-nmt}.
\end{proof}

From the evolution relation of  Proposition \ref{2pKP-cm-1}, for specific values of $i$ and $j$, various semi-discrete KP-type equations are constructed accordingly.
\begin{itemize}
	\item D$\Delta^2$pKP equation
\end{itemize}

For \eqref{2pKP-S-nmt}, let $i=j=0$, $u=S^{(0,0)}$and eliminate $S^{(0,1)}, S^{(1,0)}$, it is the  D$\Delta^2$pKP equation
\begin{equation}\label{2pKP-1}
  \partial_x(\widehat{u}-\widetilde{u})=(p-q+\widehat{u}-\widetilde{u})(u+\widehat{\widetilde{u}}-\widehat{u}-\widetilde{u}).
\end{equation}
Under transform $\mathrm{u}=u-pn-mq$, equations \eqref{2pKP-1} and  \eqref{2dkp-1} are  equivalent.

\begin{itemize}
	\item D$\Delta^2$KP equation
\end{itemize}

The equation  \eqref{2pKP-1} can be deformed as
 $[ \ln(\widehat{\mathrm{u}}-\widetilde{\mathrm{u}})]_x=\mathrm{u}+\widehat{\widetilde{\mathrm{u}}}-\widetilde{\mathrm{u}}-\widehat{\mathrm{u}},$
then, let $Q=\widehat{\mathrm{u}}-\widetilde{\mathrm{u}}$, we have
\begin{equation}\label{2dkp-1==1}
  (\ln \widetilde{Q}-\ln\widehat{ Q})_x=\widehat{\widetilde{Q}}-\widetilde{Q}-\widehat{Q}+Q,
\end{equation}
 where equation \eqref{2dkp-1} is the potential form of  \eqref{2dkp-1==1}.

\begin{itemize}
	\item D$\Delta^2$pmKP equation
\end{itemize}

For the relations  \eqref{2pKP-S-nmt}, taking $v=1-S^{(-1,0)},~w=1-S^{(0,-1)}$,
 the D$\Delta^2$pmKP equation can be obtained as
\begin{equation}\label{2pmKP-1}
 \frac{p\widehat{v}-q\widetilde{v}}{\widehat{\widetilde{v}}}- \frac{\widetilde{v}_x+pv}{\widetilde{v}}+\frac{\widehat{v}_x+qv}{\widehat{v}}=0,~ or~
 \frac{\widehat{w}_x}{\widehat{w}}-\frac{p\widetilde{w}_x-q\widehat{w}_x }{p\widetilde{w}-q\widehat{w}}+p(\frac{\widetilde{w}}{w}-\frac{\widehat{\widetilde{w}}}{\widehat{w}})=0.
\end{equation}

\begin{itemize}
	\item D$\Delta^2$SKP equation
\end{itemize}

For the relations \eqref{2pKP-S-nmt}, take $z=S^{(-1,-1)}-\frac{n}{p}-\frac{m}{q}-\frac{h}{r}$
and use the identity relation $ \frac{(vw)^{\widehat{ }}}{(vw)^{\widetilde{ }}}=\frac{(\widehat{v}w)}{(\widetilde{v}w)}\frac{(\widetilde{v}w)^{\widehat{ }}}{(\widehat{v}w){\widetilde{ }}},$
 it is the D$\Delta^2$SKP equation
\begin{equation}\label{2SKP}
 \frac{1-\widehat{z}_x}{1-\widetilde{z}_x} =\frac{(z-\widehat{z})(\widehat{z}-\widehat{\widetilde{z}})}{(z-\widetilde{z})(\widetilde{z}-\widetilde{\widehat{z}})}.
\end{equation}

\begin{itemize}
	\item  Bilinear D$\Delta^2$KP system
\end{itemize}

Definitions  the function $\tau=\tau_{n,m,x}=\det(\mathbf{I}+\mathbf{M})$ and $Y(a)=1-\mathbf{ s}^T(a+\mathbf{L})^{-1} \mathbf{u}^{(0)}$, $V(a)=1- \mathbf{w}^{(0)}(a+\mathbf{K})^{-1} \mathbf{r}$,
it can be obtained $ \frac{\widetilde{\tau}}{\tau} = Y(-p)=\frac{1}{\widetilde{V}(p)},~\frac{\widehat{\tau}}{\tau} = Y(-q)=\frac{1}{\widehat{V}(q)},~
\frac{\tau_x}{\tau} = u.$
%
Using \eqref{2pKP-S-nmt},  the following closed relationship can be obtained
\begin{equation}\label{B-DKP-1}
(p-q)\widehat{\widetilde{\tau}}\tau+\widehat{\tau}\widetilde{\tau}_x-\widetilde{\tau}\widehat{\tau}_x
 =0.
\end{equation}
Under transformation $\tau\rightarrow (p-q)^{mn+m+n}\tau$, equation \eqref{B-DKP-1} is normalized to
\begin{equation}\label{B-DKP}
  \widehat{\widetilde{\tau}}\tau+\widehat{\tau}\widetilde{\tau}_x-\widetilde{\tau}\widehat{\tau}_x
   =0.
\end{equation}

\subsection{ First-type extended systems }

\begin{prop}
For \textbf{Proposition} \ref{2pKP-cm-1}, replacing dispersion relation \eqref{2pKP-r-s} with
 \begin{equation}\label{2dpKP-t+}
   \mathbf{r}_x= \mathbf{K}\mathbf{r}+C(x)\mathbf{r},~ \mathbf{ s}_x = \mathbf{L}^T\mathbf{ s}+C^T(x)\mathbf{ s}
 \end{equation}
 and meet conditions $C(x) \mathbf{L}=\mathbf{L}C(x),~C(x)\mathbf{K}=\mathbf{K} C(x)$,
 where $C(x)=\mathbf{Diag}(c_1(x),\cdots,~ c_n(x) )$, we obtain the new relations of $x$:
\begin{equation}\label{2pKP+11}
  S^{(i,j)}_x =S^{(i,j+1)}+S^{(i+1,j)}-S^{(0,j)}S^{(i,0)}+\mathbf{w}^{(j)} [2C(x)]\mathbf{u}^{(i)}.
\end{equation}
Simultaneously,  the evolutionary relation of the direction of $n,~m$ is consistent with the propositions \ref{2pKP-cm-1}.
\end{prop}

\begin{proof}
  Similar calculation with \textbf{Proposition} \ref{2pKP-cm-1}, we can obtain the relation \eqref{2pKP+11}.
\end{proof}

\begin{itemize}
	\item First-type extended D$\Delta^2$pKP system
\end{itemize}

Making use of  \eqref{2pKP-S-n}, \eqref{2pKP-S-m}, \eqref{2pKP+11}, and \eqref{2pKP-uw-x}and taking $u=S^{(0,0)}$
and  $\Phi=(\phi_1,\phi_2,\cdots,\phi_N)$, $\Psi=(\psi_1,\psi_2,\cdots,\psi_N)^T$, with $\psi_j=\sqrt{2c_j(x)}[\mathbf{u}^{(0)}]_j$,
  $\phi_j=\sqrt{2c_j(x)}[\mathbf{w}^{(0)}]_j$,
we have
\begin{subequations}\label{2pKP+SCS1}
\begin{align}
&\partial_x(\widehat{u}-\widetilde{u})= (p-q+\widehat{u}-\widetilde{u})(u+\widehat{\widetilde{u}}-\widehat{u}-\widetilde{u})  +(E_m-E_n)\Phi\Psi,\\
&(E_n-E_m)\Psi=(p-q-\widetilde{u}+\widehat{u})\Psi ,~
(E_m-E_n)\Phi=(p-q-\widetilde{u}+\widehat{u})\widehat{\widetilde{\Phi}},\label{2pKP+SCS1-a}
\end{align}
\end{subequations}
where $E_n,~E_m$ are the shift operators in $n,m$ direction.
Under the transformation $\mathrm{u}=u-pn-mq$, the system \eqref{2pKP+SCS1},
just the first-type extended  D$\Delta^2$pKP system (ED$\Delta^2$pKP-1).
\begin{subequations}\label{2pKP=1}
  \begin{align}
&\partial_x(\widehat{\mathrm{u}}-\widetilde{\mathrm{u}})= (\widehat{\mathrm{u}}-\widetilde{\mathrm{u}})(\mathrm{u}+\widehat{\widetilde{\mathrm{u}}}-\widehat{\mathrm{u}}-\widetilde{\mathrm{u}})  +(E_m-E_n)\Phi\Psi,\label{2pKP=1-1}\\
&(E_n-E_m)\Psi=(\widehat{\mathrm{u}}-\widetilde{\mathrm{u}})\Psi ,\
  (E_m-E_n)\Phi=(\widehat{\mathrm{u}}-\widetilde{\mathrm{u}})\widehat{\widetilde{\Phi}}.
   \end{align}
\end{subequations}

\begin{itemize}
	\item First-type extended D$\Delta^2$pmKP system
\end{itemize}

From  \eqref{2pKP-S-n}, \eqref{2pKP-S-m}, \eqref{2pKP+11} and \eqref{2pKP-uw-x}, taking $v=1-S^{(-1,0)},~w=1-S^{(0,-1)}$ and $\psi_j=\sqrt{2c_j(x)}[\mathbf{u}^{(-1)}]_j,$~ $\phi_j=\sqrt{2c_j(x)}[\mathbf{w}^{(0)}]_j$, we have the D$\Delta^2$pmKP  with source
\begin{subequations}\label{2pmKP-1+scs}
  \begin{align}
   &\frac{p\widehat{v}-q\widetilde{v}}{\widehat{\widetilde{v}}}-\frac{\widetilde{v}_x+pv}{\widetilde{v}}+\frac{\widehat{v}_x+qv }{\widehat{v}}=\frac{\widetilde{\Phi}\widetilde{\Psi}}{\widetilde{v}}-\frac{  \widehat{\Phi}\widehat{\Psi}}{\widehat{v}},\\
   & (\widehat{\Phi}-\widetilde{\Phi})\widehat{\widetilde{v}}=\widetilde{\widehat{\Phi}}(p\widehat{v}-q\widetilde{v}),\
      \widetilde{\Psi}\widehat{v}-\widehat{\Psi}\widetilde{v}=\Psi(p\widehat{v}-q\widetilde{v}).
   \end{align}
\end{subequations}
Equations \eqref{2pmKP-1+scs} can be named  as the first-type extended  D$\Delta^2$pmKP system (ED$\Delta^2$pmKP-1).

\begin{itemize}
	\item First-type extended D$\Delta^2$SKP system
\end{itemize}

From  \eqref{2pKP-S-n}, \eqref{2pKP-S-m}, \eqref{2pKP+11} and \eqref{2pKP-uw-x}, taking $z=S^{(-1,-1)}-\frac{n}{p}-\frac{m}{q}-\frac{h}{r}$
and $\psi_j=\sqrt{2c_j(x)}[\mathbf{u}^{(-1)}]_j$, $\phi_j=\sqrt{2c_j(x)}[\mathbf{w}^{(-1)}]_j$,  the first-type  extended D$\Delta^2$SKP system (ED$\Delta^2$SKP-1) can be written as
\begin{subequations}\label{2SKP+scs=1}
\begin{align}
\frac{1-\widehat{z}_x+\widehat{\Phi}\widehat{\Psi}}{1-\widetilde{z}_x+\widetilde{\Phi}\widetilde{\Psi}} &=\frac{(z-\widehat{z})(\widehat{z}-\widehat{\widetilde{z}})}{(z-\widetilde{z})(\widetilde{z}-\widetilde{\widehat{z}})},\\
(\widehat{\Phi}-p\widehat{\widetilde{\Phi}})p(\widehat{z}-\widehat{\widetilde{z}})&=q(\widetilde{z}-\widehat{\widetilde{z}})(\widetilde{\Phi}-q\widetilde{\widehat{\Phi}}) ,\\
 ( \widehat{\Psi}-q\Psi)p(z-\widetilde{z})&=q(z-\widehat{z})( \widetilde{\Psi}-p\Psi).
 \end{align}
\end{subequations}

\begin{itemize}
	\item  Extended bilinear D$\Delta^2$KP system
\end{itemize}

First, for the definition of the $\tau$ function in the case of adding any function $C(x)$,
using the definitions  $Y(a)=1-\mathbf{ s}^T(a+\mathbf{L})^{-1} \mathbf{u}^{(0)}$, $V(a)=1- \mathbf{w}^{(0)}(a+\mathbf{K})^{-1} \mathbf{r}$,  it can be obtained
\begin{equation}\label{2dpKP-tt-1}
  \frac{\widetilde{\tau}}{\tau} = Y(-p)=\frac{1}{\widetilde{V}(p)},~\frac{\widehat{\tau}}{\tau} = Y(-q)=\frac{1}{\widehat{V}(q)},~
\frac{\tau_x}{\tau} = u.
\end{equation}
Using \eqref{2pKP-S-n}, \eqref{2pKP-S-m}, \eqref{2pKP+11},  the closed  relationship can be obtained as
\begin{equation}\label{2dkpscs-1}
(p-q)\widehat{\widetilde{\tau}}\tau+\widehat{\tau}\widetilde{\tau}_x-\widetilde{\tau}\widehat{\tau}_x
 =(\widehat{\tau}\widetilde{\mathbf{s}}^T-\widetilde{\tau}\widehat{\mathbf{s}}^T)(\mathbf{I}+\mathbf{M})^{-1}[C(x)-\mathbf{M}C(x)] \mathbf{u}^{(0)}\tau.
\end{equation}

Next, let $\sigma:=(\sigma_1, \cdots, \sigma_N ),~ \sigma_j=\tau A\psi_j $, i.e.$\sigma=\tau\Psi$. At the same time $\rho^*:=(\rho_1,\cdots,\rho_N),~ \rho_j=\tau  \phi_j,~ A=\frac{1}{2}(\mathbf{I}-\mathbf{M})$, i.e. $\rho^*=\tau\Phi$, where the definition of $\phi_j, ~\psi_j$
is the same as in \eqref{2pKP+SCS1-a}. Using  \eqref{2pKP-uw-x}
and $p-q-\widetilde{u}+\widehat{u}= (p-q)\frac{\widetilde{\widehat{\tau}}\tau}{ \widetilde{\tau}\widehat{\tau}}$,
it can be verified separately
\begin{subequations}
\begin{align}
 \widehat{\tau}\widetilde{\sigma}-\widetilde{\tau}\widehat{\sigma} &=(p-q)\widetilde{\widehat{\tau}}\sigma,\label{2kpscs-C1}\\
 \widetilde{\tau}\widehat{\rho}^*-\widehat{\tau}\widetilde{\rho}^* & =(p-q)\tau\widetilde{\widehat{\rho}}^*.\label{2kpscs-C2}
\end{align}
\end{subequations}
On the basis of \eqref{2dpKP-tt-1}, auxiliary functions $ \mathbf{w}(b)  =\mathbf{ s}^T (b+\mathbf{L})^{-1} (\mathbf{I}+\mathbf{M})^{-1}$, $V(a)=1- \mathbf{w}_0(a+\mathbf{K})^{-1} \mathbf{r}$ are used to obtain $\mathbf{w}(-p)=-\frac{\widetilde{\mathbf{w}}_0}{\widetilde{V}(p)}$.
Making similar processing of $\mathbf{w}(-r)$, it can be simplified to $ \mathbf{w}(-r)=-\frac{\widehat{\mathbf{w}}_0}{\widehat{V}(r)}.$
Therefore,  the source form of \eqref{2dkpscs-1} can be expressed as
\begin{equation*}
  (\widehat{\tau}\widetilde{\mathbf{s}}^T-\widetilde{\tau}\widehat{\mathbf{s}}^T)(\mathbf{I}+\mathbf{M})^{-1}[C(x)-\mathbf{M}C(x)] \mathbf{u}_0\tau_x=(p-q)\widehat{\widetilde{\rho}}^* \sigma,
\end{equation*}
with $\rho^*, ~ \sigma $ satisfy \eqref{2kpscs-C1} and \eqref{2kpscs-C2}.
Based on the above discussion and under transformation $\tau\rightarrow (p-q)^{mn+m+n}\tau$, the first-type extended bilinear D$\Delta^2$KP system can be normalized as
\begin{subequations}\label{2dkpscs}
\begin{align}
   &\widehat{\widetilde{\tau}}\tau+\widehat{\tau}\widetilde{\tau}_x-\widetilde{\tau}\widehat{\tau}_x
     =\widehat{\widetilde{\rho}}^* \sigma ,\label{2dkpscs+}\\
  & \widehat{\tau}\widetilde{\sigma}-\widetilde{\tau}\widehat{\sigma} =\widetilde{\widehat{\tau}}\sigma,~
\widetilde{\tau}\widehat{\rho}^*-\widehat{\tau}\widetilde{\rho}^* =\tau\widetilde{\widehat{\rho}}^*.\label{2kpscs=C2}
\end{align}
\end{subequations}

\subsection{ Second-type extended systems }

\begin{prop}\label{2pKP-cm-5}
For the given Sylvester equation and dispersion relation \eqref{2pKP-r-s}, using the dressed Cauchy kernel $\mathbf{M}^+=\mathbf{M}+C(m),$ the  infinite matrix and  exchange conditions can be redefined as $U_{\alpha,\beta}=\mathbf{ s}^T \mathbf{L}^\beta (\mathbf{I}+\mathbf{M}^+)^{-1}\mathbf{K}^\alpha \mathbf{r}=\mathbf{ s}^T \mathbf{L}^\beta\mathbf{u}_\alpha=\mathbf{w}_\beta\mathbf{K}^\alpha \mathbf{r},~ C(m)\mathbf{L}=\mathbf{L}C(m),~C(m)\mathbf{K}=\mathbf{K} C(m).$
 Then, the evolutionary relations can be obtained as
\begin{subequations}\label{2dpKP-U+}
\begin{align}
&\partial_x U_{\alpha,\beta} = U_{\alpha,\beta+1}+U_{\alpha+1,\beta}-U_{0,\beta}U_{\alpha,0},\label{2dpKP-U-1}\\
&p\widetilde{U}_{\alpha,\beta}=\widetilde{U}_{\alpha,\beta+1}+ pU_{\alpha,\beta}+ U_{\alpha+1,\beta}-\widetilde{U}_{\alpha,0}U_{0,\beta},\label{2dpKP-U-2}\\
&q\widehat{U}_{\alpha,\beta}=\widehat{U}_{\alpha,\beta+1}+qU_{\alpha,\beta}+ U_{\alpha+1,\beta}-\widehat{U}_{\alpha,0}U_{0,\beta}- \mathbf{w}_\beta[C(m+1)-C(m)]\widehat{ \mathbf{u}}_\alpha.\label{2dpKP-U-3}
 \end{align}
\end{subequations}
\end{prop}

\begin{proof}
Similarly, for the new dressed Cauchy kernel $\mathbf{M}^+, \mathbf{u}_\alpha, \mathbf{w}_\beta$, we have the relations
 \begin{subequations}\label{2pKP-M+}
\begin{align}
&\widetilde{\mathbf{M}}^+ =\mathbf{M}^+ +\mathbf{r}\widetilde{\mathbf{s}}^T, ~ \mathbf{M}^+_x =\mathbf{rs}^T,~
 \widehat{\mathbf{M}}^+ =\mathbf{M}^+ +\mathbf{r}\widehat{\mathbf{s}}^T+C(m+1)-C(m),\\
 &\partial_x \mathbf{u}_\alpha = \mathbf{u}_{\alpha+1}-\mathbf{u}_0U_{\alpha,0},~
 \partial_x \mathbf{w}_\beta= \mathbf{w}_{\beta+1}-\mathbf{w}_0 U_{0,\beta},\\
&\widetilde{ \mathbf{u}}_\alpha =p \mathbf{u}_\alpha + \mathbf{u}_{\alpha+1}- \widetilde{U}_{\alpha,0}\mathbf{u}_0,~
\mathbf{w}_\beta = p\widetilde{ \mathbf{w}}_\beta- \widetilde{\mathbf{w}}_{\beta+1}+ U_{0,\beta}\widetilde{\mathbf{w}}_0
\\
&\widehat{ \mathbf{u}}_\alpha =q \mathbf{u}_\alpha + \mathbf{u}_{\alpha+1}- \widehat{U}_{\alpha,0}\mathbf{u}_0- (\mathbf{I}+\mathbf{M}^+)^{-1}[C(m+1)-C(m)]\widehat{ \mathbf{u}}_\alpha.
 \end{align}
\end{subequations}
Using the relations  \eqref{2pKP-M+} and definitions $U_{\alpha,\beta}$, it can be calculated as \eqref{2dpKP-U+}.
\end{proof}

\begin{itemize}
	\item Second-type extended D$\Delta^2$pKP system
\end{itemize}

 Taking $u=U_{0,0}$ and $\Phi=(\phi_1,\phi_2,\cdots,\phi_N), ~\Psi=(\psi_1,\psi_2,\cdots,\psi_N)^T$ with $\psi_j=[\mathbf{u}_0]_j,$
$\phi_j=[c_j(m+1)-c_j(m)][\mathbf{w}_0]_j$ , from \eqref{2dpKP-U+}, we can get
\begin{subequations}\label{2pKP+SCS2}
  \begin{align}
      &\partial_x(\widehat{u}-\widetilde{u})= (p-q+\widehat{u}-\widetilde{u})(u+\widehat{\widetilde{u}}-\widehat{u}-\widetilde{u})  +(1-E_n)\Phi\widehat{\Psi},\\
      &\widetilde{\Psi}-\Psi_x = (p+u-\widetilde{u})\Psi,~
      \Phi+\widetilde{\Phi}_x = (p+u-\widetilde{u})\widetilde{\Phi}.
   \end{align}
\end{subequations}
Under the transformation  $\mathrm{u}=u-pn-mq$,  the equtaions \eqref{2pKP+SCS2}
can be rewrite to obtain second-type extended  D$\Delta^2$pKP system (ED$\Delta^2$pKP-2),
\begin{subequations}\label{2pKP=2}
  \begin{align}
&\partial_x(\widehat{\mathrm{u}}-\widetilde{\mathrm{u}})= (\widehat{\mathrm{u}}-\widetilde{\mathrm{u}})(\mathrm{u}+\widehat{\widetilde{\mathrm{u}}}-\widehat{\mathrm{u}}-\widetilde{\mathrm{u}})  +(1-E_n)\Phi\widehat{\Psi},\\
   & \widetilde{\Psi}-\Psi_x = (\mathrm{u}-\widetilde{\mathrm{u}})\Psi,~
     \Phi+\widetilde{\Phi}_x = (\mathrm{u}-\widetilde{\mathrm{u}})\widetilde{\Phi}.
   \end{align}
\end{subequations}

\begin{itemize}
	\item Second-type extended D$\Delta^2$pmKP system
\end{itemize}

From  \eqref{2dpKP-U+}, taking $v=1-U_{-1,0}, w =1-U_{0,-1} $ and  $\psi_j=[\mathbf{u}_{-1}]_j,~ \phi_j=[c_j(m+1)-c_j(m)][\mathbf{w}_0]_j$, we have
\begin{subequations}\label{scs-2mkp-2}
  \begin{align}
&p\Bigl(\frac{\widehat{v}}{\widetilde{\widehat{v}}}-\frac{v}{\widetilde{v}}\Bigr)+q\Bigl(\frac{v}{\widehat{v}}-\frac{\widetilde{v}}{\widehat{\widetilde{v}}}\Bigr)+\Bigl(\frac{\widehat{v}_x}{\widehat{v}}-\frac{\widetilde{v}_x}{\widetilde{v}}\Bigr)
  =\bigl(\frac{1}{\widehat{\widetilde{v}}}-\frac{1}{\widehat{v}}\bigr)\Phi\widehat{\Psi},\\
&(\Phi-\widetilde{\Phi}_x)\widetilde{v}=\widetilde{\Phi}(\widetilde{v}_x +pv),~
\partial_x \Psi\widetilde{v}=(\widetilde{\Psi}-p\Psi)v.
   \end{align}
\end{subequations}
Equations  \eqref{scs-2mkp-2} can be named as second-type extended  D$\Delta^2$pmKP system (ED$\Delta^2$pmKP-2)

\begin{itemize}
	\item Second-type extended D$\Delta^2$SKP system
\end{itemize}

Taking $z=U_{(-1,-1)}-\frac{n}{p}-\frac{m}{q}-\frac{h}{r}$ and $\psi_j=[\mathbf{u}_{-1}]_j,~ \phi_j=[c_j(m+1)-c_j(m)][\mathbf{w}_{-1}]_j$,  we have
\begin{subequations}\label{2SKP-2+scs=1}
  \begin{align}
    &\frac{1-\widehat{z}_x}{1-\widetilde{z}_x} =\frac{(\widehat{z}-\widehat{\widetilde{z}})\{q(z-\widehat{z})+\Phi\widehat{\Psi}\}}
  {(z-\widetilde{z})\{q(\widetilde{z}-\widetilde{\widehat{z}})+\widetilde{\Phi}\widetilde{\widehat{\Psi}}\}},\\
    &\partial_x \Phi =(p\widetilde{\Phi}-\Phi)\frac{1-\widetilde{z}_\xi}{p(z-\widetilde{z})},
     \partial_x \Psi=(\widetilde{\Psi}-p\Psi)\frac{1-z_\xi}{p(z-\widetilde{z})}.
   \end{align}
\end{subequations}
Equations \eqref{2SKP-2+scs=1} are referred to as second-type extended  D$\Delta^2$SKP system (ED$\Delta^2$SKP-2).
%
%
\section{Solutions}\label{sec-3}
To get the solution of semi-discrete KP-type equations and second-type extended semi-discrete KP-type systems, for the given matrices $(\mathbf{L},~\mathbf{K})$, we need to  define the wave factors for semi-discrete cases
\begin{subequations}\label{PWF-3.1}
\begin{align}
   &  \rho_i=(p+k_i)^n (q+k_i)^m \mathrm{e}^{\eta_i},~\eta_i=k_ix+\eta^{(0)}_i,k_i,~\eta^{(0)}_i\in \mathbb{C},\\
   &  \sigma_i=(p-l_i)^{-n}(q-l_i)^{-m}\mathrm{e}^{\zeta_i},~\zeta_i=l_ix+\zeta^{(0)}_i,l_i,~\zeta^{(0)}_i\in \mathbb{C}.
\end{align}
\end{subequations}
To get the solution of first-type extended semi-discrete KP-type systems, for the given matrices $(\mathbf{L},~\mathbf{K},~C(x))$, we need to  redefine the wave factors for semi-discrete cases
\begin{subequations}\label{PWF-3.2}
\begin{align}
   &  \rho_i=(p+k_i)^n (q+k_i)^m \mathrm{e}^{\eta_i},~\eta_i=k_ix+\int_0^{x}C(z)dz+\eta^{(0)}_i,k_i,~\eta^{(0)}_i\in \mathbb{C},\\
   &  \sigma_i=(p-l_i)^{-n}(q-l_i)^{-m}\mathrm{e}^{\zeta_i},~\zeta_i=l_ix+\int_0^{x}C(z)dz+\zeta^{(0)}_i,l_i,~\zeta^{(0)}_i\in \mathbb{C}.
\end{align}
\end{subequations}

Next, the solutions of Sylvester equation and dispersion relation can be, similarly, listed as in Ref.\cite{AML-YO-2020}.

\noindent
\textbf{Case 1.}~When $\mathbf{K}=\mathbf{D}^{[N]}(\{k_i\}^{N}_{1})$,~$\mathbf{L}=\mathbf{D}^{[N]}(\{l_i\}^{N}_{1})$,
 $C(x)=\mathbf{D}^{[N]}(\{c_i(x)\}^{N}_{1}),$ or $C(m)=\mathbf{D}^{[N]}(\{c_i(m)\}^{N}_{1}),$
the matrices
$\mathbf{r}=\mathbf{r}_{D}^{ [N]}(\{k_i\}_{1}^{N})$ and $\mathbf{s}=\mathbf{s}_{D}^{ [N]}(\{l_i\}_{1}^{N})$ are composed of
$r_i=\rho_i, s_j=\sigma_j$ respectively, where
\begin{equation}\label{M-1}
\mathbf{M}  = \Bigl(\frac{r_i s_j}{k_i+l_j}\Bigr)_{N\times N}\, .
\end{equation}

\noindent
\textbf{Case 2.}~When $ \mathbf{K}=\mathbf{J}^{[N]}(k_1), ~\mathbf{L}=\mathbf{J}^{[N]}(l_1),$ ~$C(x)=\mathbf{J}^{[N]}(c_1(x))$ or $C(m)=\mathbf{J}^{[N]}(c_1(m))$,
we have $\mathbf{r}=\mathbf{r}_{J}^{ [N]}(k_1)$ and $\mathbf{s}=\mathbf{s}_{J}^{[N]}(l_1)$ composed of
$r_i=\frac{\partial^{i-1}_{k_1}\rho_1}{(i-1)!},\  s_j=\frac{\partial^{N-j}_{l_1}\sigma_1}{(N-j)!}$ respectively,
and $\mathbf{M}=\mathbf{FGH}$
where
\begin{equation}\label{FGH-2}
\mathbf{F}=\mathbf{T}^{[N]}(\{r_i\}^{N}_{1}),~~\mathbf{G}=\mathbf{G}^{[N;N]}_{JJ}(k_1;l_1),~~ \mathbf{H}=\mathbf{H}^{[N]}(\{s_j\}^{N}_{1}).
\end{equation}

\noindent
\textbf{Case 3.}~When
\begin{subequations}
\label{Ga,Lb-gen}
\begin{align}
&\mathbf{K}=\mathrm{Diag} \{\mathbf{D}^{[N_1]}(\{k_i\}^{N_1}_{1}),
\mathbf{J}^{[N_2]}(k_{N_1+1}),\mathbf{J}^{[N_3]}(k_{N_1+2}),\cdots,
\mathbf{J}^{[N_s]}(k_{N_1+(s-1)})\},\\
&\mathbf{L}=\mathrm{Diag} \{\mathbf{D}^{[N_1]}(\{l_i\}^{N_1}_{1}),
\mathbf{J}^{[N_2]}(l_{N_1+1}), \mathbf{J}^{[N_3]}(l_{N_1+2}),\cdots,
\mathbf{J}^{[N_s]}(l_{N_1+(s-1)})\},
\end{align}
\end{subequations}
with $\sum_{i=1}^sN_i=N$, $C(x)=\mathrm{Diag}\{\mathbf{D}^{[N_1]}(\{c_i(x)\}^{N_1}_{1}), \mathbf{J}^{[N_2]}(c_{N_1+1}(x)),
\cdots, \mathbf{J}^{[N_2]}(c_{N_1+(s-1)}(x))\}$ or $C(m)=\mathrm{Diag}\{\mathbf{D}^{[N_1]}(\{c_i(m)\}^{N_1}_{1}), \mathbf{J}^{[N_2]}(c_{N_1+1}(m)),
\cdots, \mathbf{J}^{[N_2]}(c_{N_1+(s-1)}(m))\}$, we have
\begin{equation}
\mathbf{r}=\left(
\begin{array}{l}
\mathbf{r}_{D}^{[N_1]}(\{k_i\}_{1}^{N_1})\\
\mathbf{r}_{J}^{[N_2]}(k_{N_1+1})\\
~~\vdots\\
\mathbf{r}_{J}^{[N_s]}(k_{N_1+(s-1)})
\end{array}
\right),~~~\mathbf{s}=\left(
\begin{array}{l}
\mathbf{s}_{D}^{[N_1]}(\{l_j\}_{1}^{N_1})\\
\mathbf{s}_{J}^{N_2}(l_{N_1+1})\\
~~\vdots\\
\mathbf{s}_{J}^{[N_s]}(l_{N_1+(s-1)})
\end{array}
\right),
\end{equation}
where
\begin{align}
&\mathbf{F}=\mathrm{Diag}\{
\mathbf{D}^{[N_1]}(\{r_i\}^{N_1}_{1}),
\mathbf{T}^{[N_2]}(k_{N_1+1}),\mathbf{T}^{[N_3]}(k_{N_1+2}),\cdots,
\mathbf{T}^{[N_s]}(k_{N_1+(s-1)}) \},\label{KP-r-M-g-F}\\
&\mathbf{H}=\mathrm{Diag}\{
\mathbf{D}^{[N_1]}_{D}(\{s_j\}^{N_1}_{1}),
\mathbf{H}^{[N_2]}(l_{N_1+1}),
\mathbf{H}^{[N_3]}(l_{N_1+2}),
\cdots,
\mathbf{H}^{[N_s]}(l_{N_1+(s-1)})\},\label{KP-r-M-g-H}
\end{align}
with $\mathbf{M}=\mathbf{FGH}$, and $\mathbf{G}$ is a block matrix
$\mathbf{G}=(\mathbf{G}_{i,j})_{s\times s}$,
where
\[
\begin{array}{l}
\mathbf{G}_{1,1}=
    \Bigl(\frac{\rho_i \sigma_j}{k_i+l_j}\Bigr)_{N_1\times N_1};~~
 \mathbf{G}_{1,j}=\mathbf{G}^{[N_1;N_j]}_{DJ}(\{k_i\}^{N_1}_{1};l_{N_1+j-1}),~(1<j\leq s); \\
\mathbf{G}_{i,1}=\mathbf{G}^{[N_i;N_1]}_{JD}(k_{N_1+i-1};\{l_j\}^{N_1}_{1}),~(1<i\leq s);
~~ \mathbf{G}_{i,j}=\mathbf{G}^{[N_i;N_j]}_{JJ}(k_{N_1+i-1};l_{N_1+j-1}),~(1<i,j\leq s).
\end{array}
\]

Note that the commutative conditions are satisfied for the above three cases.

\section{Conclusion}\label{sec-4}
Semi-discrete versions of soliton systems, i.e. systems given by integrable partial  differential-difference equations, has significant roles in mathematics and physics. Based on KP-type equations involving one continuous variable and two discrete variables, we derive the D$\Delta^2$pKP, D$\Delta^2$pmKP and D$\Delta^2$SKP equations, as well as two kinds of corresponding extended systems. The two types of extensions are given in terms of  squared eigenfunctions. This also demonstrates, to some extent, the capability of the Cauchy matrix method in constructing new integrable systems.

For KP-type equation with one continuous variable and two discrete variables, we can continue to study its conservation law and related properties. At the same time, the extended Volterra chain system can be further investigated in the future. Furthermore, for the extended D$^2\Delta$KP equation, the squared eigenfunction  was given as a symmetry of the D$^2\Delta$KP equation in \cite{CK-JNMP-2017}. Whether the two types of squared eigenfunction extensions  can be utilized as symmetries is worth exploring.
The research on the relationships between integrable systems and symmetries \cite{MFM-2015-DKP-Ghost-sym,YYQ-JPA-2009}  not only sheds new light on the current theory of integrable systems, but also enhances many areas of pure mathematics.

\vskip 15pt
\subsection*{Acknowledgments}
This project is supported by Natural Science Foundation of Henan (No.242300421687) and Henan Provincial Science and Technology Research Project (No.222102210258).
\vskip15pt
\subsection*{Conflict of Interest} The authors of this work declare that they have no conflicts of interest.

\end{document}